\documentclass[12pt]{article}
\usepackage[usenames]{color}
\usepackage{graphicx, subfigure}
\usepackage{amsmath, amsthm, amssymb}
\usepackage{amsfonts}
\usepackage{fullpage}
\usepackage{ifthen}
\usepackage{url}
\usepackage[sort&compress]{natbib}
\usepackage{multirow}
\usepackage{rotating}

\usepackage[linesnumbered,algoruled,boxed,lined,commentsnumbered]{algorithm2e}

\makeatletter
\def\@biblabel#1{}
\makeatother

\theoremstyle{plain}
\newtheorem{prop}{Proposition}

\theoremstyle{definition}
\newtheorem{example}{Example}

\theoremstyle{remark}

\title{Analytic models for active shooter incidents with civilian resistance}

\author{Liang Hong\footnote{Liang Hong is a Professor in the Department of Mathematical Sciences,  The University of Texas at Dallas, 800 West Campbell Road, Richardson, TX 75080, USA. Tel.:~972-883-2161. Email address: liang.hong@utdallas.edu.}}
\date{\today}

\begin{document}

\maketitle

\begin{abstract}
The number of active shooter incidents in the US has been increasing alarmingly.  It is imperative for the government as well as the public to understand these events.  Though both analytic and agent-based models have been proposed for studying active shooter incidents, there are only a few analytic models in the literature,  and none incorporate civilian resistance.  This article analytically investigates the survival probability of a civilian during an active shooter incident when he can hide, fight, or run, depending on whether or not he is in a closed arena and whether or not he is armed.  The key findings are (i) the civilian's chance of survival decreases over time,  irrespective of his action; (ii) there are desperate situations in which a civilian should fight back even if he is unarmed; (iii) carrying a firearm does not always increase the civilian's chance of survival during an active shooter incident; (iv) carrying a firearm makes ``resistance'' more likely to be the optimal action of the civilian; (v) ``hide'' might be the best action even if the civilian is armed; (vi) more armed civilians might increase a civilian's chance of survival, but it is not always the case.
\smallskip

\end{abstract}

\section{Introduction}
The active shooter incident has become a regular event in our society.  It has happened in all kinds of public places, such as churches,  schools,  supermarkets, theaters, and shopping malls, just to name a few.  The FBI defines an active shooter as one or more individuals actively engaged in killing or attempting to kill people in a populated area (e.g., FBI 2024).  The rate of active shooter incidents in the US has been increasing at an alarming rate in the past two decades (e.g.,  FBI 2016, 2019,  2023, 2024; Soni and Tekin 2020, The Violence Project 2020, Wiki 2023).  These violent events  lead to the deaths of multiple victims, cause lasting trauma to survivors, and instill fear in other people.  Multiple government agencies, including the Department of Homeland Security,  the Department of Justice,  the FBI,  and many local police departments, have been working tirelessly to find better ways to prevent and respond to active shooter incidents.  In academia,  scholars have been studying active shooter incidents for years.  In terms of research goals,  the literature primarily focuses on identifying key factors in mass shooting events, such as access to weapons,  domestic violence, legislation, police response,  shooter profile, etc; see, for instance,  Gius (2015),  Blum and Jaworski (2017),  Cabrera and Kwon (2018),  Kwon and Cabrera (2019),  Silva and Greene-Colozzi (2021), Fridel (2021), Geller et al. (2021). As far as research methods are concerned,  some authors rely on quantitative studies; see, for example, Lee (2013),  Kingshott and McKenzie (2013), Lemieux (2014),  Metzl and MaclEish (2015),  Dong (2019), Ihatsu et al. (2022).  In particular,  Lemieux (2014) used cross-section analysis to argue that both gun culture and firearm laws have significant impacts on deaths by guns. Dong (2019) applied qualitative content analysis to find evidence that depictions of firearms in mass media coverage of mass shootings generally constituted firearms characteristics that produced feelings of fear in the intended audience.  Ihatsu et al. (2022) performed an inductive content analysis on several police pre-trial investigation records and reports by the investigation commission to determine attributes of violence resembling terrorism by approaching the subject from the position of first responders.

However, most articles on active mass shootings seem to be based on quantitative studies and employ some types of quantitative models.  There are three major kinds of quantitative models in the literature.  The first kind is a statistical model.  Usually, such a statistical model is fitted using real-world data to gain insight into active shooter incidents.  For example,  Lankford (2016) applied analysis of variance to empirically demonstrate that fame-seeking rampage shooters have existed for decades and have become more common in recent years and that these offenders tend to kill and wound significantly more victims.  Lee et al. (2018) used a Bayesian Poisson regression model to study the casualty rate during an active shooter event with three parameters: civilian evacuation time, the response of police,  and firearm discharge by the shooter and police.  They found that immediate evacuation would decrease the casualty rate,  and so would the rapid and well-directed deployment of first responders with high-precision discharge abilities.  Lei and MacKenzie (2024) proposed a hierarchical model to quantify the risk of mass shootings at specific locations.  Cao et al. (2024) applied regression analysis to develop a model of mass shooting severity; their studies suggested that factors, such as having more firepower at the shooting scenes, lack of employment, being married, and having more mental health issues, all increase casualties in mass shootings.  

The second kind is an agent-based simulation model.  This seems to be the dominating model in the literature. 
For example,  Hayes et al. (2014) applied an agent-based simulation model to study the potential effectiveness of banning assault weapons and high-capacity magazines.  According to their results, banning assault weapons will not affect the number of people killed or wounded in a mass shooting; however, a ban on high-capacity magazines will result in a small reduction in the number of deaths during a mass shooting.   Anklam et al. (2015) used an agent-based simulation model to find that a decrease in the authority's response time leads to reduced victim casualties.  Gunn et al.  (2017) proposed a solution to optimizing guidance delivery for active shooter events.   Steward (2017) employed an agent-based simulation model to study how enforcement response time, civilian response strategy, and cognitive delay of the civilians might affect the expected number of casualties in a classroom shooting.  She found that a slower police response time and having all civilians hide can increase the expected number of civilian casualties. Glass et al. (2018) proposed an agent-based simulation model to investigate the means of mitigating the risk of a single active shooter incident within a large event venue by increasing the number of armed off duty law enforcement officers available to respond to the incident.  Their results show a negative logarithmic relationship between the number of civilians shot and the number of off-duty law enforcement officers available.  Lu et al. (2023) proposed an agent-based model that incorporates both the killing force and counterforce; their model suggested that there is an optimal response time for the police to minimize civilian casualties. 

The third kind is an analytic model.  There is a paucity of analytic models in the literature. One exception is Sporrer (2020), who applied several analytic models to study active shooter events in different types of arenas: closed, open, and complex.  So far no analytic models have accounted for civilian resistance.  The purpose of this article is to investigate the survival probability of a civilian during an active shooter event when the civilian can fight, hide, or run, depending on whether or not the civilian is in a closed arena and whether or not the civilian is armed.  Similar to Sporrer (2020), we will study active shooter incidents for the aforementioned three kinds of arenas.  Different from Sporrer (2020) and other previous authors who focused on the expected number of civilian casualties,  we are mainly interested in the best action of a civilian that maximizes his probability of survival and we allow ``fight'' as a possible action.  We consider the case of a single armed civilian and the case of multiple armed civilians for each type of arena.   To preview the key results,  the following findings will be established across different scenarios:
\begin{enumerate}
\item[(i)]A civilian's chance of survival decreases over time regardless of the action he chooses.
\item[(ii)]There are some desperate situations in which a civilian, armed or not, should fight back.  This confirms the prediction made in Section~5.2 of Sporrer (2020). 
\item[(iii)]Carrying a firearm might increase a civilian's chance of survival in some cases. However,  carrying a firearm does not necessarily increase a civilian's chance of survival in every case. 
\item[(iv)]Everything else being equal, carrying a firearm  makes ``resistance'' more likely to be the best action of a civilian. 
\item[(v)]There are some circumstances where ``hide'' is the best option even if a civilian is armed and the authority cannot arrive quickly. 
\item[(vi)]A civilian's chance of survival is not necessarily an increasing function of the number of armed civilians. Indeed, when the number of armed civilians exceeds a threshold,  a civilian's chance of survival will decrease due to the increased likelihood of ``friendly fire''.
\end{enumerate}
These findings can be applied by a variety of audiences. For example,  law enforcement agencies can learn from (i) that a quick response to an active shooter incident is critical.  Parents who are looking for schools for their children can also gain insight from (i): the distance between the school and the nearest police station is an important variable to consider especially when a school is located in a rural area.  For another example, (v) dispels the misconception in some folks' minds that they should always fight back during an active shooter incident if they are armed.

Finally,  a remark on the choice of our terminology  ``Fight/Hide/Run''. This is similar to the term ``Avoid/Barricade/Counter'' used in the ABC active shooter training at many workplaces, except that the former is slightly more general than the latter.  If the arena is closed (e.g., a school or office building), they are the same.  However, if the arena is open like a theater or concert hall,  ``hide'' will be a more appropriate term than ``barricade''. 

The remainder of this article is organized as follows.  In Section~2, we study the case where there is only one possibly armed civilian.  We will derive the survival probabilities of a civilian over time and use them to decide the optimal action of a civilian.  Next,  in Section~3, we investigate the case where there are multiple possibly armed civilians.   Finally,  we conclude the article with some remarks in Section~4.  

\section{One possibly armed civilian}

In this section, we assume there is one civilian of focus who may or may not be armed.  No other civilians might fight back.  In addition, the following common assumptions are taken throughout this section. 

\begin{enumerate}
\item[(a)]There is only one active shooter. He shoots indiscriminately into the crowd unless he meets resistance.  He will always kill those who fight first if he meets resistance.   
\item[(b)]The authority will arrive in $T_2$ minutes, where $T_2$ is a positive integer.  When the authority arrives,  the shooter will kill himself or surrender himself to the authority; otherwise, the authority will neutralize the shooter immediately.
\item[(c)]All civilians are rational decision-makers. They make their decisions without collaborating with others. Each of them wants to maximize their own chance of survival.
\item[(d)]At the beginning of each minute,  a civilian makes his decision for his action during that minute and sticks to it until the beginning of the next minute.
\item[(e)]All civilians will be hit by the shooter with equal probabilities. 
\item[(f)]A bullet strikes at most one person. That is, there are no secondary casualties.
\end{enumerate}

\subsection{Closed arena}
Similar to Sporrer (2020),   a closed arena is understood to be an area that does not offer people the opportunity to exit.  A typical example is a classroom during a campus shooting.  When such an event occurs,  the civilian of focus can find himself either already in a closed arena or in the hallway where one can quickly enter a closed arena.   Let us first examine the case where a civilian is already in the closed arena. 

\subsubsection{Already in the closed arena}
This is the situation where a student or teacher finds himself in a classroom when a campus shooting occurs.  In this case,  we also assume the following:
\begin{enumerate}
\item[(i)]There are $N$ civilians in the closed arena, where $N$ is a positive integer.  The active shooter is not considered as a civilian.
\item[(ii)]The shooter is already at the only door of the closed arena.  There is no other place to exit the room.  Therefore,  the civilian of focus can either fight or hide.  A civilian who does not fight will be killed by the shooter with equal probability.  The shooter will walk around until every civilian is dead. 
\item[(iii)]The shooter is equipped with a rifle, and his killing rate is $m$ people per minute where $m$ is a positive integer and $N=T_1 m$ for some positive integer $T_1$.   He carries plenty of magazines that are far more than enough to kill all the $N$ civilians.  His rifle will not malfunction during the shooting.  He plans to kill as many civilians as possible.
\item[(iv)]If an unarmed civilian fights back, his chance of success will be $p_1$. If an armed civilian fights back, his probability of success will be $p_2$, where $p_2>p_1>0$. 
\end{enumerate}
To avoid repetitions in our presentation, we will use $p$ to denote the survival probability of the civilian of focus, where $p=p_1$ if he is unarmed and $p=p_2$ if he is armed. \\

It is clear from our Assumption (iii) that $T_1$ is the number of minutes the shooter needs to kill all the civilians in the arena.  Thus, $T_1$ plays a vital role in our analysis.  In particular, we need to consider two cases: (I)~$T_1>T_2$ and (II)~$T_1\leq T_2$.  When $T_1>T_2$, the authority will arrive before the shooter can kill all the civilians in the arena; hence, every civilian should stand a nonzero chance of survival anytime.  However, if $T_1\leq T_2$, the shooter will murder all the civilians before the authority arrives; as a result, the chance of survival of each civilian should be zero after $T_1$ minutes.  Our analysis below confirms this intuition. \\

\noindent \textbf{Case I: $T_1>T_2$}\\

\noindent For the first minute, if the civilian fights, his survival probability will be $p$; if he hides, his chance of survival will be $(1-m/N)$.  For the latter probability, note that the shooter will kill $m$ civilians in the first minute and each of the $N$ civilians is equally likely to be among these $m$ unfortunate victims.  It follows that the civilian's chance of being killed during the first minute is ${N-1 \choose m-1}/{N \choose m}=m/N$,  implying the chance of survival is $(1-m/N)$.  Note that if the civilian fights,  he will either end the shooting or get himself killed right away.  However, if he hides, he may still choose to fight or hide in the next minute provided he has survived the previous minute. Therefore,  if the civilian fights during the second minute, his chance of survival equals
\begin{equation}
P(\text{succeed in fighting back})\times P\{\text{survived the first minute}\}=p\left(1-\frac{m}{N}\right).
\end{equation}
But if he decides to hide again during the second minute, his survival probability is 
\begin{eqnarray}
& & P(\text{survive the second minute $\mid$ survived the first minute})\times P\{\text{survived the first minute}\}  \nonumber \\
&=& \left[\frac{(N-m)-m}{N-m}\right] \left(1-\frac{m}{N}\right)=1-\frac{2m}{N}.
\end{eqnarray}
Following the same line of reasoning, we can derive the civilian's survival probabilities at other times up to $T_2$ minutes. 
After $T_2$ minutes,  the authority arrives. Thus, the civilian is surely to live regardless of his action,  provided he has survived the previous $T_2$ minutes.  That is,  the conditional probability of survival given that the civilian has survived the first $T_2$ minutes is $1$.  Therefore, the survival probability of the civilian beyond $T_2$ minutes is
\begin{equation}
P(\text{survived the first $T_2$ minutes})\times 1=\left(\frac{N-T_2m}{N}\right)\times 1=1-\frac{T_2m}{N},
\end{equation}
where $1$ stands for the civilian's conditional probability of survival given that he has survived the first $T_2$ minutes.\\

\begin{table}[!th]
\begin{center}
\begin{tabular}{cccccccc}
\hline
minute & $1$ &  $2$ & $\ldots$  & $T_2$ & $T_2+1$ & $\ldots$ & $T_1$\\
\hline
fight & $p$ & $p\left(1-\frac{m}{N}\right)$   & $\ldots$  & $p\left[1-\frac{(T_2-1)m}{N}\right]$  & $1-\frac{T_2m}{N}$ & $\ldots$  & $1-\frac{T_2m}{N}$\\
hide & $1-\frac{m}{N}$ & $1-\frac{2m}{N}$ & $\ldots$  &  $1-\frac{T_2m}{N}$ & $1-\frac{T_2m}{N}$ & $\ldots$  & $1-\frac{T_2m}{N}$ \\
\hline
\end{tabular}
\end{center}
\caption{Survival probabilities of the civilian of focus in a closed arena when $T_1>T_2$, where the $i$-th  cell of the ``fight'' row is the survival probability of the civilian when he hides for the first $(i-1)$ minutes and starts to fight back at the beginning of the $i$-th minute.}
\label{table:sp1.1}
\end{table}

Table~\ref{table:sp1.1} provides the civilian's survival probabilities over time, where the $i$-th  cell of the ``fight'' row is his survival probability when he hides for the first $(i-1)$-th minutes and starts to fight back at the beginning of the $i$-th minute.  Table~\ref{table:sp1.1} offers some insight into civilian resistance during an active shooter incident.  First,  if the civilian decides to hide, his chance of survival decreases over time until the authority arrives.  The same can be said for the civilian if he plans to fight back. Since $N=T_1m$ and $T_2<T_1$, we know that $1-T_2/m>0$.  That is,  if the authority is able to arrive quickly,  a civilian will have a positive chance of survival if he decides to hide.  This means the authority response time is critical in an active shooter event.  For parents who are considering new schools for their children, they might want to take this into consideration.  Everything else being equal,  a school that is close to a police station or staffed with an armed guard should be chosen over the one that is not.  Also,  some schools in rural areas are usually far away from the authority, they should consider adding an armed guard if they do not have one yet.

Next,  if $p>(1-m/N)$, then the civilian should fight back at the inception of the shooting.  More generally, if  $p[1-(k-1)m/N]>(1-km/N)$ for some $1\leq k\leq T_2$ but $p[1-(j-1)m/N]<(1-jm/N)$ for $1\leq j\leq k-1$,  then a civilian should fight back during the $k$-th minute; otherwise,  ``hide'' will be a better option.   For an unarmed civilian (i.e., $p=p_1$), this analytically confirms the prediction made in Sporrer (2020) that there are some desperate situations where the best option of an unarmed civilian is to fight the shooter.  This shows that hiding is not always the only rational option for an unarmed civilian. On the other hand,  if $p[1-(k-1)m/N]<(1-km/N)$ for all $1\leq k\leq T_2$, or, equivalently,  $p<[N-km]/[N-(k-1)m]=1-m/[N-(k-1)m]$ for all $1\leq k\leq T_2$,  then a civilian should hide all the way through.  Note that $1-m/[N-(k-1)m]$ is decreasing in $k$; therefore,  the condition $p<1-m/[N-(k-1)m]$ for all $1\leq k\leq T_2$ is equivalent to the condition $p<1-m/(N-m)$.

Moreover, if $p_2>1-m/[N-(k-1)m]$ for some $1\leq k\leq T_2$, then $p_2>\max\{p_1, 1-m/[N-(k-1)m]\}$, showing that carrying a firearm may increase a civilian's chance of survival in some cases.  However, carrying a firearm does not always increase one's chance of survival.  For example, if  $p_2<1-m/[N-(k-1)m]$ for all $1\leq k\leq T_2$, then a civilian would better hide throughout even if he is armed.

Finally,  it is possible that $p_2>1-m/[N-(k-1)m]$ but $p_1<1-m/[N-(k-1)m]$ for some $1\leq k\leq T_2$.  In this case, carrying a firearm increases the odds that ``fight'' is the best option.  \\

\noindent \textbf{Case II: $T_1\leq T_2$}\\

By similar analysis as in the case where $T_1>T_2$, we can obtain the civilian's survival probabilities over time that are given by Table~\ref{table:sp1.2}.

\begin{table}[!th]
\begin{center}
\begin{tabular}{ccccccccc}
\hline
minute & $1$ &  $2$ & $\ldots$  & $T_1$-1 & $T_1$ & $T_1+1$ & $\ldots$ & $T_2$\\
\hline
fight & $p$ & $p\left(1-\frac{m}{N}\right)$ & $\ldots$ & $p\left[1-\frac{(T_1-2)m}{N}\right]$ & $p\left[1-\frac{(T_1-1)m}{N}\right]$  & $0$  & $\ldots$& $0$\\
hide & $1-\frac{m}{N}$ & $1-\frac{2m}{N}$ & $\ldots$  & $1-\frac{(T_1-1)m}{N}$ &  $0$ & $0$ & $\ldots$  & $0$ \\
\hline
\end{tabular}
\end{center}
\caption{Survival probabilities of the civilian of focus in a closed arena when $T_1 \leq T_2$, where the $i$-th  cell of the ``fight'' row is the survival probability of the civilian when he hides for the first $(i-1)$ minutes and starts to fight back at the beginning of the $i$-th minute.}
\label{table:sp1.2}
\end{table}

A few remarks in order.  First,  if the civilian wants to hide and wait for the authority, his chance of survival will decrease over time and become zero before the authority arrives.  The same can be said for the civilian if he plans to fight back.  That is,  if the authority is unable to arrive quickly (e.g., a campus shooting in a rural area),  the civilian will stand no chance of surviving in a closed arena shooting after a certain time point, regardless of whether he decides to hide and wait for the authority.

Second,  if $p>(1-m/N)$, then a civilian should fight back right away.  More generally, if  $p[1-(k-1)m/N]>(1-km/N)$ for some $1\leq k\leq T_1$ but $p[1-(j-1)m/N]<(1-jm/N)$ for all $1\leq j\leq k-1$,  then a civilian should fight back during the $k$-th minute; otherwise,  he had better to hide.  Again, the aforementioned prediction made in Sporrer (2020) is confirmed here analytically.   It is worth underlining that even if $p[1-(k-1)m/N]<(1-km/N)$ for all $1\leq  k\leq T_1-1$, or, equivalently,  $p<1-m/[N-(k-1)m]$ for all $1\leq k\leq T_1-1$,  the civilian should  fight in the $T_1$-th minute because his chance of survival is zero if he hides during the $T_1$-th minute.  

We see again that,  if $p_2>1-m/[N-(k-1)m]=\max\{p_1, 1-m/[N-(k-1)m]\}$ for some $1\leq k\leq T_2$, then carrying a firearm will increase a civilian's chance of survival.  Once again, this does not imply that carrying a firearm will necessarily increase one's chance of survival because it is possible that $p_2<1-m/[N-(k-1)m]$ for all $1\leq k\leq T_1$.  Hence, an armed civilian should not always choose to fight. If $p_2>1-m/[N-(k-1)m]$ but $p_1<1-m/[N-(k-1)m]$ for some $1\leq k\leq T_1$,  then carrying a firearm will increase the chance of civilian resistance when everything else being equal.  

\subsubsection{Hallway interaction}
A typical example of this is when a student or teacher finds himself in the hallway when a campus shooting happens.  Besides the common assumptions in Section~2, we also assume the following:
\begin{enumerate}
\item[(1)]There are $M$ civilians in the hallway, where $M$ is a positive integer.  
\item[(2)]There are $K$ separate closed arenas close to each civilian. The shooter will search these $K$ closed arenas one by one to look for his victims. The shooter will not re-enter a closed arena he has searched.  Once the shooter enters a closed arena, he will not move to another closed arena until he kills all civilians in the current arena.
\item[(3)]Each civilian can either run or hide in one of the $K$ closed arenas. In addition, the civilian of focus can also fight.
\item[(4)]If an unarmed civilian fights, his chance of success will be $\tilde{p}_1$.  If an armed civilian fights back, his probability of success will be $\tilde{p}_2$, where $\tilde{p}_2>\tilde{p}_1>0$.  Note that $\tilde{p}_1$ may or may not equal to $p_1$. Similarly,  $\tilde{p}_2$ may or may not equal to $p_2$.   We will  use $\tilde{p}$ to denote the survival probability of a civilian, where $\tilde{p}=\tilde{p}_1$ if he is unarmed and $\tilde{p}=\tilde{p}_2$ if he is armed. 
\item[(5)]If a civilian runs, he can safely exit the area with a probability of $p_r>0$.
\item[(6)]If a civilian decides to hide in one of the $K$ nearby closed arenas, he will succeed with a probability  $p_h>0$. 
\item[(7)]The hallway interaction ends within a minute.
\item[(8)]The authority will not arrive during hallway interaction. 
\item[(9)]There will be $N$ civilians in each of the $K$ closed arenas once the hallway interaction ends, where $N<M/K$.
\end{enumerate}
The decision of the civilian of focus during hallway interaction (i.e., the first minute) is easy.  If $p_h>\max \{p_r, \tilde{p}\}$, then he should hide into a closed arena; if $p_r>\max \{p_h, \tilde{p}\}$, then he should run;  if $\tilde{p}>\max \{p_r,  p_h\}$, then he should fight back.   

At the beginning of the second minute, the shooter may enter the closed arena where the civilian of focus is hiding with a probability $1/K$.  If this happens, then the results in Section~2.1.1 will apply from the second minute to the $\max\{T_1+1, T_2+1\}$-th minute.   However, if the shooter enters an arena other than where the civilian of focus is hiding at the beginning of the second minute,  which happens with a probability $1-/K^2$, then the civilian of focus will hide from the second minute until the $(T_1+1)$-th minute.  In this case, if $T_2\leq T_1+1$, the civilian of focus will surely survive.  However, if $T_2>T_1+1$, then at the beginning of the $(T_1+2)$-th minute   we will be in the same situation as the one we face at the beginning of the second minute,  except now $K$ is being replaced by $K-1$.  Then, the shooter might enter the arena where the civilian is hiding with a probability $1/(K-1)$, in which case the results in Section~2.1.1 will apply again from the $(T_1+2)$-th minute until the  $\max\{2T_1+1, T_1+T_2+1\}$-th minute.  The analysis for the rest of the time until the $(KT_1+1)$-th minute is similar.  The larger $K$ is, the longer the whole event will last and hence the more likely the civilian of focus will have a nonzero chance of survival as time passes.

\subsection{Open arena}
By an open area, we mean an area that allows people to exit in an unconstrained manner. A typical example is an outdoor area for a festival.  Besides the common assumptions in Section~2,  we make the following assumptions:
\begin{enumerate}
\item[(1)]There are $N$ civilians in the open arena, where $N$ is a positive integer.
\item[(2)]The shooter is equipped with a rifle, and his killing rate is $m$ people per minute, where $m$ is a positive integer. The shooter carries plenty of magazines that are far more than enough to kill all the $N$ civilians.  His rifle will not malfunction during the shooting.  He plans to kill as many civilians as possible. 
\item[(3)]The crowd exits the open arena at a rate of $e$ people per minute, where $e$ is a positive integer and $e<N-m$.  For each minute,  none of the $e$ civilians who exit the arena is among the $m$ civilians who are killed during that minute. 
\item[(4)]There is a positive integer $T_1$ such that $N=T_1 (e+m)$ for some positive integer $T_1$.  
\item[(5)]Every civilian can hide or run.  In addition, the civilian of focus can also fight.  The shooter will hit anyone who does not fight with equal probability.
\item[(6)]If the civilian of focus runs,  he can safely exit the arena in a minute with a probability $p_r>0$. That is,   with a probability $p_r$, he will be one of the $e$ civilians who can safely exit the arena.
\item[(7)]If an unarmed civilian fights, his chance of success will be $p_1$.  If an armed civilian fights back, his probability of success will be $p_2$, where $p_2>p_1>0$.  We still use $p$ to denote the survival probability of a civilian when he fights back, where $p=p_1$ if he is unarmed and $p=p_2$ if he is armed. 
\end{enumerate}

\subsubsection{Case I: $T_1>T_2$}
In this case, the incident will not last beyond time $T_2$.  Note that if a civilian chooses to run, he can either exit the area or get killed in a minute.  If he fights back, he will either subdue the shooter with a probability $p$ or get killed immediately.  But if he hides, he can still choose to fight,  hide,  or run in the next minute.  Employing a similar argument as in Section~2.1.1, we can obtain the survival probabilities of the civilian of focus over time.  These probabilities are summarized in .   For example,   the civilian is still in the area at the beginning of the second minute if and only if he chooses to hide in the first minute.  This happens with a probability $1-m/(N-e)$ because, according to our assumption, none of the $e$ civilian who exit the open arena is shot during the first minute.  It follows that there are $N-(m+e)$ civilians left in the arena at the end of the first minute.  If he chooses to hide again at the beginning of the second minute, then his chance of surviving during the second minute would be 
\begin{equation}
\left(1-\frac{m}{N-e}\right)\left[1-\frac{m}{N-(e+m)-e}\right]=\prod_{k=1}^2\left[1-\frac{m}{(N-e)-(k-1)(e+m)}\right].
\end{equation}

\begin{table}[!th]
\begin{center}
\begin{tabular}{ccccc}
\hline
minute & $1$ &  $2$ & $\ldots$  & $T_2$   \\
\hline
fight & $p$ & $p\left(1-\frac{m}{N-e}\right)$ & $\ldots$ & $p\prod_{k=1}^{T_2-1}\left[1-\frac{m}{(N-e)-(k-1)(e+m)}\right]$     \\
hide & $1-\frac{m}{N-e}$ & $\prod_{k=1}^2\left[1-\frac{m}{(N-e)-(k-1)(e+m)}\right]$ & $\ldots$ &  $\prod_{k=1}^{T_2}\left[1-\frac{m}{(N-e)-(k-1)(e+m)}\right]$   \\
run &  $p_r$ & $p_r\left(1-\frac{m}{N-e}\right)$ & $\ldots$ & $p_r\prod_{k=1}^{T_2-1}\left[1-\frac{m}{(N-e)-(k-1)(e+m)}\right]$   \\
\hline
\end{tabular}
\end{center}
\caption{Survival probabilities of the civilian of focus in an open arena when $T_1>T_2$, where the $i$-th  cell of the ``fight'' (respectively ``run'') row is the survival probability of the civilian when he hides for the first $(i-1)$ minutes and starts to fight back (respectively ``run'') at the beginning of the $i$-th minute.}
\label{table:sp2.1}
\end{table}
Note that the assumptions on $e, m$, and $N$ imply that all probabilities in Table~\ref{table:sp2.1} are positive.
Based on Table~\ref{table:sp2.1}, we can make almost the same observations as for the case $T_1>T_2$ in Section~2.1.1. For the sake of brevity, they will not be repeated here.  The same will be done for the remainder of the paper unless new insight arises.  Readers will see that Observations~(i)--(vi) at the end of Section~1 will hold in all cases.

\subsubsection{Case II: $T_1\leq T_2$}

In this case, the event will end at time $T_1$ if not before.  Similar to the case where $T_1>T_2$,  we can obtain the survival probabilities of the civilian over time; see Table~\ref{table:sp2.2}. 

\begin{table}[!th]
\begin{center}
\begin{tabular}{ccccccccc}
\hline
minute & $1$ &  $2$ & $\ldots$ & $T_1-1$ & $T_1$ & $T_1+1$  & $\ldots$ & $T_2$\\
\hline
fight & $p$ & $pA_1$ & $\ldots$ & $p\prod_{k=1}^{T_1-2} A_k$ & $p\prod_{k=1}^{T_1-1}A_k$  & $0$  & $\ldots$ & 0 \\
hide & $A_1$ & $\prod_{k=1}^2A_k$ & $\ldots$ &  $\prod_{k=1}^{T_1-1}A_k$ & $0$   & $0$   & $\ldots$ & 0\\
run &  $p_r$ & $p_rA_1$ & $\ldots$ & $p_r\prod_{k=1}^{T_1-2} A_k$  & $p_r\prod_{k=1}^{T_1-1}A_k$  & $0$   & $\ldots$ & 0\\
\hline
\end{tabular}
\end{center}
\caption{Survival probabilities of the civilian of focus in an open arena when $T_1\leq T_2$, where $A_k=\left[1-\frac{m}{(N-e)-(k-1)(e+m)}\right]$ and the $i$-th  cell of the ``fight'' (respectively ``run'') row is the survival probability of the civilian when he hides for the first $(i-1)$ minutes and starts to fight back (respectively ``run'') at the beginning of the $i$-th minute.}
\label{table:sp2.2}
\end{table}

Similar to the case where $T_1>T_2$,  it is easy to see that all probabilities in Table~\ref{table:sp2.2} are positive.  In addition,  all survival probabilities after the $(T_1-1)$-th minute  in the ``hide'' row are zeros because $N=T_1 (e+m)$, i.e.,  all civilians will have either exited the arena or been killed at the end of the $T_1$-th minute.  Again,  we can make almost the same observations as in Section~2.1.2. For example,  if the civilian of focus chooses to hide during the first $T_1-1$ minutes, then he should either fight back or run during the $T_1$-th minute. 

\subsection{Complex arena}
Following Sporrer (2020), we take a complex arena to be an area that allows people to escape through only a finite number of exits.  Most mass shootings in a theater or shopping mall fall into this category.  Similar to Sporrer (2020), let us partition the process of such an incident into three phases: 
\begin{enumerate}
\item[]Phase $0$: initiation,
\item[]Phase $1$: maximum capacity, 
\item[]Phase $2$: seek and shoot.
\end{enumerate}
Phase $0$ is the first few moments after the shooting starts.  Civilians during Phase $0$ often fail to recognize what is going on due to confusion; hence, they do not fight, hide, or run.  Phase $1$ is the time interval during which civilians realize the nature of the incident and many of them rush to escape the arena through the few exits.  People cannot exit the arena freely during Phase $1$ because the highest exit rate of each exit is reached.  Phase $2$ is the last stage of the incident when the density of evacuees is low enough to allow them to exit the arena freely.  Since the majority of civilians will have either exited the arena or been killed,  the shooter needs to seek out potential targets during Phase $2$.  Thus, his killing rate drops during Phase $2$.	

Besides the common assumptions in Section~2,  let us also take the following assumptions:
\begin{enumerate}
\item[(1)]There are $N$ civilians in the complex arena, where $N$ is a positive integer. 
\item[(2)]Due to confusion, no one fights,  hides,  or runs during Phase 0.
\item[(3)]The shooter will kill $N-N_1$ people during Phase 0, where $N_1$ is a positive integer such that $N>N_1$.  His killing rates during Phase 1 and Phase 2 are $m_1$ people per minute and $m_2$ people per minute respectively, where $m_1$ and $m_2$ are positive integers and $m_1>m_2$.  The shooter is equipped with a rifle and plenty of magazines that are far more than enough to kill all the $N$ civilians.  
\item[(4)]Those who decide to hide during Phase 1 will not fight or run during that phase. 
\item[(5)]Phase 1 lasts $n$ minutes. The crowd exits the arena at a rate of $e_1$ people per minute during Phase 1, where $e_1$ is a positive integer.  During each minute of Phase 1, none of the $e_1$ civilians who exit the arena is killed. 
\item[(6)]If an unarmed civilian fights back during Phase 1, his chance of success will be $\tilde{p}_1$. If an armed civilian  fights back during Phase 1, his probability of success will be $\tilde{p}_2$ where $\tilde{p}_2>\tilde{p}_1>0$.  We will use $\tilde{p}$ to denote the survival probability of a civilian where $\tilde{p}=\tilde{p}_1$ if he is unarmed and $\tilde{p}=\tilde{p}_2$ if he is armed.   
\item[(7)]The crowd exits the area at rate of $e_2$ people per minute during Phase 2, where $e_2$ is a positive integer. For each minute of Phase 2, none of the $e_2$ civilians who exit the arena is killed. 
\item[(8)]If an unarmed civilian fights back during Phase 2, his chance of success will be $p_1$. If an armed civilian fights back during Phase 2, his probability of success will be $p_2$ where $p_2>p_1>0$.  We will use $p$ to denote the survival probability of a civilian who fights back, where $p=p_1$ if he is unarmed and $p=p_2$ if he is armed. 
\item[(9)]If a civilian runs during Phase 1,  he can safely exit the arena in a minute with a probability of $p_r$.  If a civilian runs during Phase 2 he can safely exit the arena in a minute with a probability of $\tilde{p}_r$. 
\item[(10)]Let $N_2=N_1-n(e_1+m_1)$ be the number of people in the complex arena at the beginning of Phase 2.  We assume $N_2=T_1(e_2+m_2)$ for some positive integer $T_1$.  
\item[(11)]The shooter seeks and shoots during phase 2 until every remaining civilian in the arena is dead.
\end{enumerate}

During Phase 0, the crowd size goes from $N$ to $N_1$. Hence, the survival probability of each civilian is $N_1/N$. 
The survival probabilities of the civilian during Phase $1$ are the same as those probabilities in Table~\ref{table:sp2.1}, with appropriate substitutions.  For example,  if the civilian chooses to hide for the first $i-1$ minutes but fights back during the $i$-th minute of Phase 1, then his survival probability is
\[
\frac{\tilde{p} N_1}{N} \prod_{k=1}^{i-1}  \left[1-\frac{m_1}{(N-e_1)-(k-1)(e_1+m_1)}\right], \quad \text{$i=1, \ldots, n$}.
\]

\subsubsection{When $T_1>T_2$}
Note that if a civilian is still in the arena and alive at the beginning of Phase 2, he must have survived Phase 0 and Phase 1 and his action in Phase 1 must be ``hide''.  Thus, the probability that the civilian is alive in the arena at the beginning of Phase 2, denoted as $p_s$, is given by
\begin{equation}
\label{eq:survprob}
p_s=\frac{N_1}{N}\prod_{k=1}^{n}  \left[1-\frac{m_1}{(N-e_1)-(k-1)(e_1+m_1)}\right].
\end{equation}

Following the same line of reasoning as in Section~2.2,  we can obtain a civilian's probabilities of survival over time during Phase 2. Table~\ref{table:sp3.1} summarizes these probabilities.  We again make similar observations as in Section~1.2.1 and Section~2.2.1. 

\begin{table}[!th]
\begin{center}
\begin{tabular}{ccccc}
\hline
minute & $1$ &  $2$ & $\ldots$  & $T_2$   \\
\hline
fight & $p_s \tilde{p}$ & $p_s \tilde{p}\left(1-\frac{m_2}{N_2-e_2}\right)$ & $\ldots$ & $p_s \tilde{p}\prod_{k=1}^{T_2-1}\left(1-\frac{m_2}{N'_k}\right)$     \\
hide & $p_s \left(1-\frac{m_2}{N_2-e_2}\right)$ & $p_s\prod_{k=1}^2 \left(1-\frac{m_2}{N_k'}\right)$ & $\ldots$ &  $p_s \prod_{k=1}^{T_2}\left(1-\frac{m_2}{N_k'} \right)$   \\
run &  $p_s \tilde{p}_r $ & $p_s \tilde{p}_r \left(1-\frac{m_2}{N_2-e_2}\right)$ & $\ldots$ & $ p_s \tilde{p}_r\prod_{k=1}^{T_2-1} \left(1-\frac{m_2}{N_k'}\right)$   \\
\hline
\end{tabular}
\end{center}
\caption{Survival probability of the civilian of focus in a complex arena during Phase 2 when $T_1>T_2$,  where $p_s$ is given by (\ref{eq:survprob}),  $N'_k=(N_2-e_2)-(k-1)(e_2+m_2)$, and the $i$-th  cell of the ``fight'' (respectively ``run'') row is the survival probability of the civilian when he hides for the first $(i-1)$ minutes and starts to fight back (respectively ``run'') at the beginning of the $i$-th minute.}
\label{table:sp3.1}
\end{table}

\subsubsection{When $T_1\leq T_2$}
Similar to the case where $T_1<T_2$,  we can derive all the survival probabilities in Table~\ref{table:sp3.2}.  Observations in this case are completely analogous to those in Section~1.2.2 and Section~2.2.2.

\begin{table}[!th]
\begin{center}
\begin{tabular}{ccccccccc}
\hline
minute & $1$ &  $2$ & $\ldots$ & $T_1-1$ & $T_1$   & $T_1+1$ & $\ldots$ & $T_2$\\
\hline
fight & $p_s \tilde{p}$ & $p_s \tilde{p}B_1$ & $\ldots$ & $p_s\tilde{p}\prod_{k=1}^{T_1-2}B_k$  & $p_s\tilde{p}\prod_{k=1}^{T_1-1}B_k$   &$0$  & $\ldots$ & 0\\
hide & $p_s B_1$ & $p_s \prod_{k=1}^2 B_k$ & $\ldots$ &  $p_s \prod_{k=1}^{T_1-1}B_k$ & $0$  & $0$ & $\ldots$ & 0\\
run &  $p_s \tilde{p}_r $ & $p_s \tilde{p}_r B_k$ & $\ldots$ & $ p_s \tilde{p}_r\prod_{k=1}^{T_1-2} B_k$ & $ p_s \tilde{p}_r\prod_{k=1}^{T_1-1} B_k$ & $0$ & $\ldots$ & 0 \\
\hline
\end{tabular}
\end{center}
\caption{Survival probability of the civilian of focus in a complex arena during Phase 2 when $T_1 \leq T_2$, where $p_s$ is given by (\ref{eq:survprob}),  $B_k=\left(1-\frac{m_2}{N'_k} \right)=\left[1-\frac{m_2}{((N_2-e_2)-(k-1)(e_2+m_2))}\right]$, and the $i$-th cell of the ``fight'' (respectively ``run'') row is the survival probability of the civilian when he hides for the first $(i-1)$ minutes and starts to fight back (respectively ``run'') at the beginning of the $i$-th minute.}
\label{table:sp3.2}
\end{table}

\section{Multiple possibly armed civilians}

In this section, we examine the case where there are multiple civilians of focus who may or may not carry a firearm and no other civilians will fight back. In such a case,  it will be shown that the survival probability of a civilian of focus is not necessarily an increasing function of the number of armed civilians.  Intuitively, when ``too many'' civilians have guns and start to shoot, it will cause confusion and lead to civilian deaths.  An extreme case is when the active shooter is dressed in ordinary clothes and carries a handgun, and every civilian in a closed area is armed with a handgun too.  Then an armed civilian is unlikely to identify who the shooter is.  Also,  any armed civilian can mistakenly kill another civilian when he is fighting the shooter.  This is even more true when the civilian is under the stress of life and death.  A key finding in this section is that,  as the number of armed civilians increases, the survival probability of a civilian of focus in this case may first increase and then decrease. Throughout this section, we assume 
\begin{enumerate}
\item[(1)]There is only one active shooter. He shoots indiscriminately into the crowd unless he meets resistance.  He will always kill those who fight back first if he meets resistance.   
\item[(2)]The authority will arrive in $T_2$ minutes, where $T_2$ is a positive integer.  When the authority arrives,  the shooter will kill himself or surrender himself to the authority; otherwise, the authority will neutralize the shooter immediately.
\item[(3)]All civilians are rational decision-makers. They make their decisions without collaborating with others. Each of them wants to maximize their own chance of survival.
\item[(4)]At the beginning of each minute,  a civilian makes his decision for his action during that minute and sticks to it until the beginning of the next minute.
\item[(5)]All civilians will be hit by the shooter with equal probabilities. 
\item[(6)]A bullet strikes at most one person. That is, there are no secondary casualties.
\item[(7)]If an unarmed civilian fights back, no other civilians, armed or not, will follow him.  If an armed civilian fights back, no unarmed civilians will follow him. 
\item[(8)] For $i\geq 1$, the number of survived civilians and armed civilians at the beginning of $i$-th minute are $N_i$ and $K_i$ respectively, where $K_i\leq N_i$,  $N=N_1\geq N_2\geq\ldots $,  and $N\geq K_1\geq K_2\geq \ldots$.  If one of the armed civilians fights back at the beginning of the $i$-th minute,  then $j_{K_i}-1$ of them will follow immediately,  where $1\leq j_{K_i}\leq K_i$ and $j_{K_i}\leq j_{K_{i-1}}$.  
\end{enumerate}

\subsection{Closed arena}

\subsubsection{Already in the closed arena}
Besides the common assumptions for Section~3,  we also take the following assumptions in this section:
\begin{enumerate}
\item[(a)]There are $N$ civilians in the closed arena, where $N$ is a positive integer.  The active shooter is not considered as a civilian.
\item[(b)]The shooter is already at the only door of the closed arena.  There is no other place to exit the room.  Therefore,  a civilian of focus can either fight or hide.  A civilian who does not fight will be killed by the shooter with equal probability.  The shooter will walk around until every civilian is dead. 
\item[(c)]The shooter is equipped with a rifle, and his killing rate is $m$ people per minute, where $m$ is a positive integer,   and $K_i-K_{i-1}<m$, and $N=T_1 m$ for some positive integer $T_1$.   He carries plenty of magazines that are far more than enough to kill all the $N$ civilians.  His rifle will not malfunction during the shooting.  He plans to kill as many civilians as possible.
\item[(d)]When an armed civilian fights back, he will kill the shooter with a probability $p_2$.  At the same time, he may also mistakenly kill another civilian with a probability of $p_f$, where $p_f=cp_2$ for some $0<c<1$.  
\item[(e)]If an unarmed civilian fights back, he will subdue the shooter with a probability $p_1$, where $p_2>p_1>0$. 
\end{enumerate}

Let us look at the $i$-th minute. First, the probability that no civilian is mistakenly killed when an armed civilian fights back equals $(1-p_f)^{j_{K_i}}$. Next,  the event that an armed civilian fights back and survives the shooting is the same as the event that the shooter is killed by an armed civilian or the remaining $j_{K_i}-1$ armed civilians who follow him and, at the same, he is not mistakenly killed by the other $j_{K_i}-1$ armed civilians.  The probability that the shooter is killed by one of the $j_{K_i}$ armed civilians is $1-(1-p_2)^{j_{K_i}}$. The probability that an armed civilian is not killed by another armed civilian when he fights back is $[1-p_f/(N_i-1)]^{j_{K_i}-1}$ where $N_i=N-(i-1)m$. 
Hence, 
\begin{eqnarray}
\label{eq:prob1}
g(K_i )& \equiv& \mathbb{P}\{\text{an armed civilians fights back and survives}\} \nonumber \\
 &  = &[1-(1-p_2)^{j_{K_i}}][1-p_f/(N_i-1)]^{j_{K_i}-1},
\end{eqnarray}
where  $K_i\in \{1, 2, \ldots, N_i\}$.

\begin{prop}
\label{prop:closed}
Suppose $j_{K_i}=K_i$ and $\frac{[1-p_f/(N_i-1)]} {[(1-p_f/(N_i-1))(1-p_2)]^{(1-p_2)^2}} >1$.  Then the function $g$, given by (\ref{eq:prob1}), first increases and then decreases.
\end{prop}

\begin{proof}
Let $f(x)=[1-(1-p_2)^x][1-p_f/(N_i-1)]^{x-1}$ where $x\geq 1$. Then $f(x)$  is proportion to $[1-p_f/(N_i-1)]^x-[(1-p_2)(1-p_f/(N_i-1))]^x$. Therefore, 
\begin{eqnarray}
f'(x) &\propto& [1-p_f/(N_i-1)]^x\ln [1-p_f/(N_i-1)]\\
        & &-\left[(1-p_2)(1-p_f/(N_i-1))]^x \ln[(1-p_2)(1-p_f/(N_i-1))\right]  \nonumber \\ 
	&=& [1-p_f/(N_i-1)]^x \left[\ln(1-p_f/(N_i-1))-(1-p_2)^x\ln[(1-p_f/(N_i-1))(1-p_2)] \right] \nonumber \\
         &=& [1-p_f/(N_i-1)]^x\ln \left[ \frac{(1-p_f/(N_i-1))} {[(1-p_f/(N_i-1))(1-p_2)]^{(1-p_2)^x}} \right].
\end{eqnarray}
Therefore, $f'(x)>0$ if and only if 
\begin{equation}
\frac{[1-p_f/(N_i-1)]} {[(1-p_f/(N_i-1))(1-p_2)]^{(1-p_2)^x}} >1.
\end{equation}
In particular,  $f'(2)>0$ if and only if 
\begin{equation}
\frac{[1-p_f/(N_i-1)]} {[(1-p_f/(N_i-1))(1-p_2)]^{(1-p_2)^2}} >1.
\end{equation}
Therefore, the proposition follows.
\end{proof}
\noindent \textbf{Remark.} The condition $\frac{[1-p_f/(N_i-1)]} {[(1-p_f/(N_i-1))(1-p_2)]^{(1-p_2)^2}} >1$ is relatively mild. For example, it is satisfied when $p_2=0.45$,  $p_f=0.2$, and $N_i=20$; see Figure~\ref{fig:prob} for a plot of $g$ in this case.

\begin{figure}[!th]
\begin{center}
\scalebox{0.5}{\includegraphics{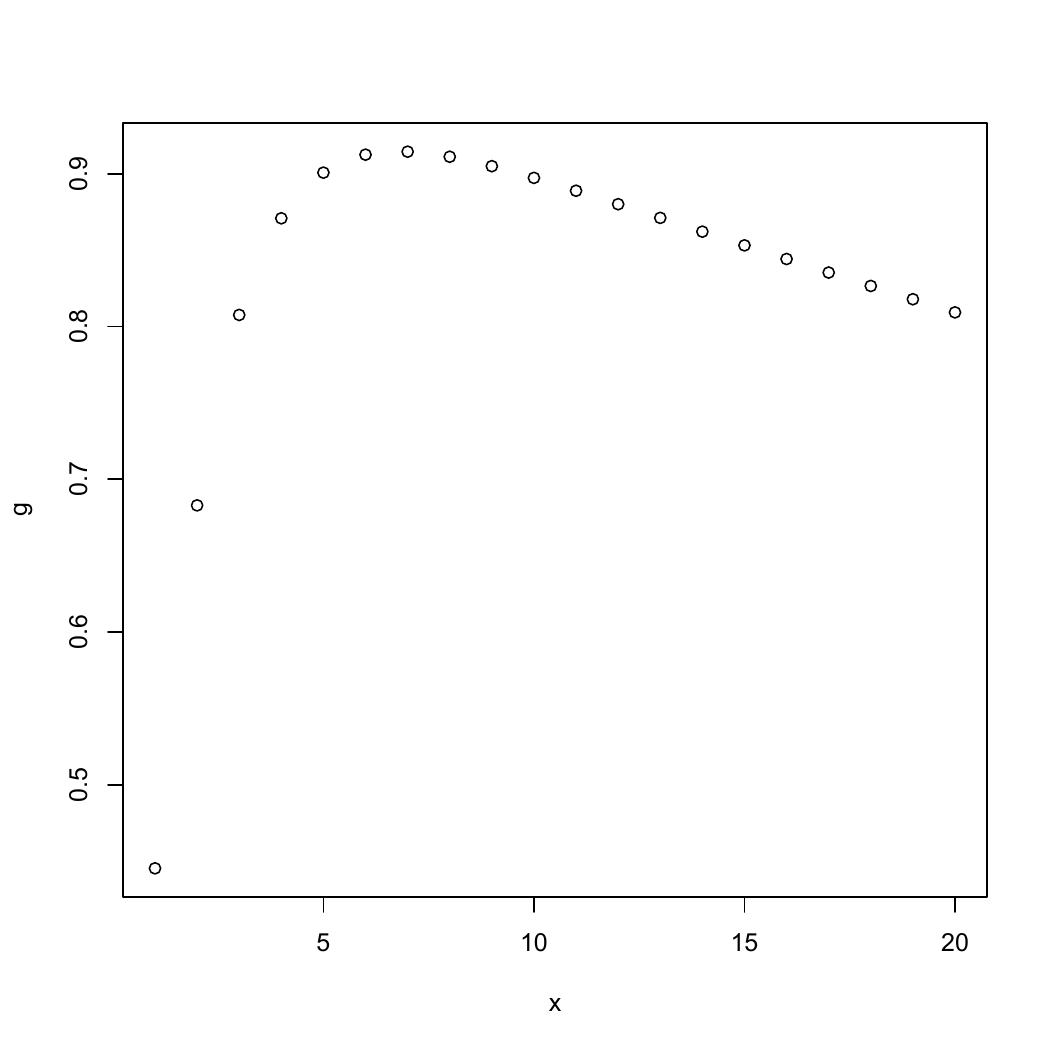}}
\end{center}
\caption{Plot of $g(K_i)$ when $j_{K_i}=K_i$, $p_2=0.45$,  $p_f=0.2$, and $N_i=20$.}
\label{fig:prob}
\end{figure}

Proposition~\ref{prop:closed} shows that more armed civilians do not necessarily increase the survival probability of a civilian of focus. Quite the opposite,  when there are ``too many'' armed civilians, the survival probability of a civilian of focus starts to decrease due to the increased likelihood of ``friendly fire''.  An interesting practical implication of this is that the scenario where everyone carries a firearm might be worse than the scenario where only a few of them has a firearm. \\

For $i=1, \ldots, \max\{T_1, T_2\}$, we put $p_{2i}=[1-(1-p_2)^{j_{K_i}}][1-p_f/(N_i-1)]^{j_{K_i}-1}$ if a civilian is armed and $p_{2i}=p_1$ if he is unarmed.   Similar to our analysis in Section~2.1, we can derive the survival probabilities of a civilian over times for both $T_1>T_2$ and  $T_1\leq T_2$.   We summarize the results in Table~\ref{table:sp4.1} and Table~\ref{table:sp4.2}  respectively.  The respective observations for both cases are similar to those in Section~2.1.1. For example,  in the case where $T_1>T_1$, if  $p_{2T_i}[1-(i-1)m/N]>(1-im/N)$ for some $1\leq i \leq T_2$ but $p_{2j}[1-(j-1)m/N]<(1-jm/N)$ for $1\leq j\leq i-1$,  then a civilian of focus should fight back during the $i$-th minute; otherwise,  ``hide'' will be a better option.  On the other hand,  if $p_{2i}[1-(i-1)m/N]<(1-im/N)$ for all $1\leq i\leq T_2$, or, equivalently,  $p_{2i}<[N-im]/[N-(i-1)m]=1-m/[N-(i-1)m]$ for all $1\leq i\leq T_2$,  then a civilian of focus should hide.  
\begin{table}[!th]
\begin{center}
\begin{tabular}{cccccccc}
\hline
minute & $1$ &  $2$ & $\ldots$  & $T_2$ & $T_2+1$ & $\ldots$ & $T_1$\\
\hline
fight & $p_{21}$ & $p_{22}\left(1-\frac{m}{N}\right)$   & $\ldots$  & $p_{2T_2}\left[1-\frac{({T_2}-1)m}{N}\right]$  & $1-\frac{{T_2}m}{N}$ & $\ldots$  & $1-\frac{{T_2}m}{N}$ \\
hide & $1-\frac{m}{N}$ & $1-\frac{2m}{N}$ & $\ldots$  &  $1-\frac{{T_2}m}{N}$ & $1-\frac{{T_2}m}{N}$ & $\ldots$ & $1-\frac{{T_2}m}{N}$  \\
\hline
\end{tabular}
\end{center}
\caption{Survival probabilities of a civilian of focus in a closed arena when $T_1>T_2$, where the $i$-th  cell of the ``fight'' row is the survival probability of the civilian when he hides for the first $(i-1)$ minutes and starts to fight back at the beginning of the $i$-th minute.}
\label{table:sp4.1}
\end{table}

\bigskip

\begin{table}[!th]
\begin{center}
\begin{tabular}{cccccccc}
\hline
minute & $1$ &  $2$ & $\ldots$  & $T_1-1$ & $T_1$ & $T_1+1$ & $\ldots$\\
\hline
fight & $p_{21}$ & $p_{22}\left(1-\frac{m}{N}\right)$   & $\ldots$ & $p_{2(T_1-1)}\left[1-\frac{(T_1-2)m}{N}\right]$ & $p_{2T_1}\left[1-\frac{(T_1-1)m}{N}\right]$  & $0$ & $\ldots$  \\
hide & $1-\frac{m}{N}$ & $1-\frac{2m}{N}$ & $\ldots$  &  $1-\frac{(T_1-1)m}{N}$ & $0$ & $0$ & $\ldots$   \\
\hline
\end{tabular}
\end{center}
\caption{Survival probabilities of a civilian in a closed arena when $T_1\leq T_2$,  where the $i$-th  cell of the ``fight'' row is the survival probability of the civilian when he hides for the first $(i-1)$ minutes and starts to fight back at the beginning of the $i$-th minute.}
\label{table:sp4.2}
\end{table}

\subsubsection{Hallway interaction}
As for the hallway interaction, the assumptions, analysis, and observations are similar to those in Section~2.1.2 except that $p_2$  is replaced by $p_{21}$.  It is easy to see again that the survival probability of a civilian of focus, as a function of the number of armed civilians, first increases and then decreases.

\subsection{Open arena}
The first six assumptions in Section~2.2 will be taken here along with Assumption~(a) in Section~3.1.  The survival probabilities of a civilian of focus are the same as in Section~2.2 except that $p$ is replaced by $p_{2i}$ in all expressions for $1\leq i \leq \max \{T_1, T_2\}$; see Table~\ref{table:sp5.1} and Table~\ref{table:sp5.2}.  Our observations are the same as in Section~2.2.  Again, we can easily see that a civilian's survival probability,  as a function of the number of armed civilians,  first increases and then decreases. 

\begin{table}[!th]
\begin{center}
\begin{tabular}{ccccc}
\hline
minute & $1$ &  $2$ & $\ldots$  & $T_2$   \\
\hline
fight & $p_{21}$ & $p_{22}\left(1-\frac{m}{N-e}\right)$ & $\ldots$ & $p_{2T_1}\prod_{k=1}^{T_2-1}\left[1-\frac{m}{(N-e)-(k-1)(e+m)}\right]$     \\
hide & $1-\frac{m}{N-e}$ & $\prod_{k=1}^2\left[1-\frac{m}{(N-e)-(k-1)(e+m)}\right]$ & $\ldots$ &  $\prod_{k=1}^{T_2}\left[1-\frac{m}{(N-e)-(k-1)(e+m)}\right]$   \\
run &  $p_r$ & $p_r\left(1-\frac{m}{N-e}\right)$ & $\ldots$ & $p_r\prod_{k=1}^{T_2-1}\left[1-\frac{m}{(N-e)-(k-1)(e+m)}\right]$   \\
\hline
\end{tabular}
\end{center}
\caption{Survival probabilities of a civilian of focus in an open arena when $T_1>T_2$,  where  the $i$-th  cell of the ``fight'' (respectively ``run'') row is the survival probability of the civilian when he hides for the first $(i-1)$ minutes and starts to fight back (respectively ``run'') at the beginning of the $i$-th minute.}
\label{table:sp5.1}
\end{table}

\begin{table}[!th]
\begin{center}
\begin{tabular}{ccccccc}
\hline
minute & $1$ &  $2$ & $\ldots$  & $T_1-1$ & $T_1$ & $T_1+1$  \\
\hline
fight & $p_{21}$ & $p_{22}A_1$ & $\ldots$ & $p_{2(T_1-1)}\prod_{k=1}^{T_1-2}A_k$ & $p_{2T_1}\prod_{k=1}^{T_1-1}A_k$  & $0$   \\
hide & $1-\frac{m}{N-e}$ & $\prod_{k=1}^2A_k$ & $\ldots$ &  $\prod_{k=1}^{T_1-1}A_k$   & $0$ & $0$   \\
run &  $p_r$ & $p_rA_1$ & $\ldots$ & $p_r\prod_{k=1}^{T_1-2}A_k$ & $p_r\prod_{k=1}^{T_1-1}A_k$  & $0$   \\
\hline
\end{tabular}
\end{center}
\caption{Survival probabilities of a civilian in an open arena when $T_1\leq T_2$,  where $A_k=\left[1-\frac{m}{(N-e)-(k-1)(e+m)}\right]$ and the $i$-th  cell of the ``fight'' (respectively ``run'') row is the survival probability of the civilian when he hides for the first $(i-1)$ minutes and starts to fight back (respectively ``run'') at the beginning of the $i$-th minute..}
\label{table:sp5.2}
\end{table}

\begin{example}
Suppose $N=210$, $m=15$, $e=25$, $p_r=0.1$, $p_1=0.05$, $p_2=0.3$,  $c=1/3$ (i.e., $p_f=0.1$),  $K_i=\lfloor N_i/4\rfloor$,  $j_{K_i}=\lfloor N_i/20\rfloor$,  and $T_2=4$, where $\lfloor a\rfloor$ is the greatest integer less than or equal to $a$.  In this case, $T_1=N/(e+m)=5$ and Table~\ref{table:sp5.1} specializes to Table~\ref{table:example1.1}.

\begin{table}[!th]
\begin{center}
\begin{tabular}{ccccccccc}
\hline
minute & $1$ &  $2$ & $3$ & $4$  & $5$   & $6$ & $\ldots$\\
\hline
fight (armed) & $0.968$ &  $0.863$ & $0.724$ &  $0.534$ & $0.276$ & $0.276$ &$\ldots$ \\
fight (unarmed) & $0.050$ & $0.046$ & $0.041$ & $0.035$  & $0.027$ & $0.027$ & $\ldots$ & \\
hide & $0.919$ & $0.824$ & $0.706$ & $0.543$ &  $0.217$ &  $0.217$& $\ldots$\\
run &  $0.100$ &  $0.092$ & $0.083$ & $0.071$ & $0.054$ & $0.054$ & $\ldots$\\
\hline
\end{tabular}
\end{center}
\caption{Survival probabilities in Table~\ref{table:sp5.1}  when  $N=210$, $m=15$, $e=25$, $p_r=0.1$, $p_1=0.05$, $p_2=0.3$,  $p_f=0.1$, $K_i=N_i/4$,  $j_{K_i}=N_i/20$,  and $T_2=4\leq T_1=5$.}
\label{table:example1.1}
\end{table}


\begin{figure}[!th]
\begin{center}
\includegraphics[scale=0.45]{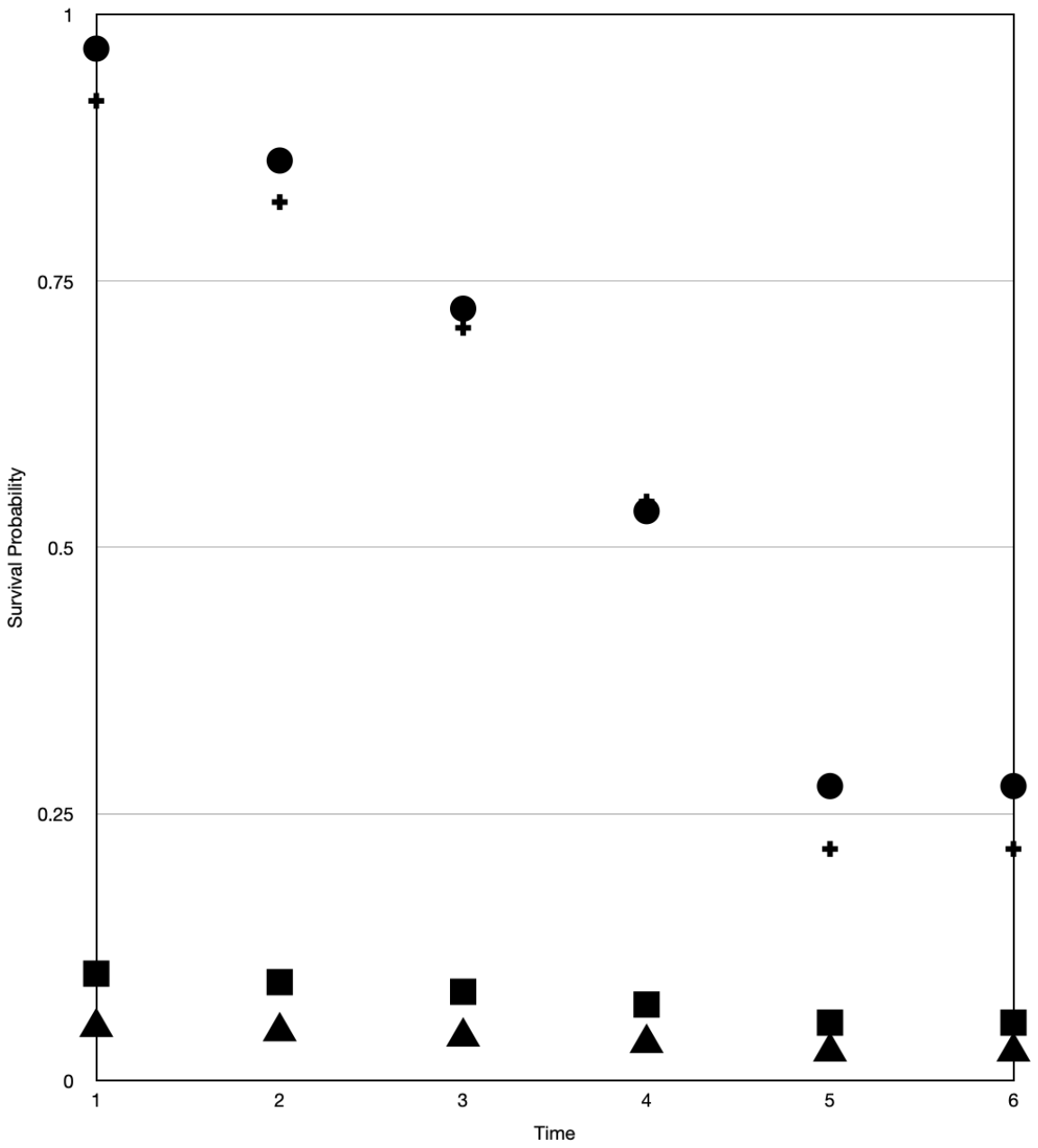}
\end{center}
\caption{Plots of the survival probabilities in Table~\ref{table:example1.1}, where the circle, triangle,  plus sign, and rectangle, symbols denote ``fight (armed)'', ``fight (unarmed)'', ``hide'', and ``run''  rows respectively.  .}
\label{fig:survprob}
\end{figure}

Figure~\ref{fig:survprob} displays a scatter plot of the survival probabilities in Table~\ref{table:example1.1}.  Clearly, if a civilian of focus is unarmed, he should hide throughout.  However, if he is armed, the should always fight except during the fourth minute where he should hide.  Table~\ref{table:example1.2} is the special case of Table~\ref{table:sp2.1} under the same setup when there is exactly one possibly armed civilian of focus.  Comparing the first row of Table~\ref{table:example1.1} with that of Table~\ref{table:example1.2}, we see that it might be better to fight earlier when there are more armed civilians, and more armed civilian will significantly increase the survival chance of an armed civilian because some armed civilians  will join him when he fights back.  This is different from the situation depicted by Figure~\ref{fig:prob} where more armed civilians is not necessary better.  Table~\ref{table:example1.2} also shows that the civilian should hide throughout. Therefore,  ``hide'' might be the best strategy when there is only one armed civilian, and when there are more armed civilians, ``fight'' might be the optimal action.

\begin{table}[!th]
\begin{center}
\begin{tabular}{ccccccccc}
\hline
minute & $1$ &  $2$ & $3$ & $4$  & $5$   & $6$ & $\ldots$\\
\hline
fight (armed) & $0.300$ &  $0.276$ & $0.249$ &  $0.213$ & $0.162$ & $0.162$ &$\ldots$ \\
fight (unarmed) & $0.050$ & $0.046$ & $0.041$ & $0.035$  & $0.027$ & $0.027$ & $\ldots$ & \\
hide & $0.919$ & $0.824$ & $0.706$ & $0.543$ &  $0.217$ &  $0.217$& $\ldots$\\
run &  $0.100$ &  $0.092$ & $0.083$ & $0.071$ & $0.054$ & $0.054$ & $\ldots$\\
\hline
\end{tabular}
\end{center}
\caption{Survival probabilities of a civilian of focus in an open arena when $N=210$, $m=15$, $e=25$, $p_r=0.1$, $p_1=0.05$, $p_2=0.3$, and $T_2=4\leq T_1$ when there is only one possibly armed civilian.}
\label{table:example1.2}
\end{table}



\end{example}

\subsection{Complex arena}
In this section,  we take all assumptions in Section~2.3 except that 
Assumption~(6) there will be replaced by
\begin{enumerate}
\item[(6')]If an unarmed civilian fights back during Phase 2, his chance of success will be $\tilde{p}_1$. If an armed civilian  fights during Phase 2, his probability of success will be $\tilde{p}_2$ where $\tilde{p}_2>\tilde{p}_1>0$.  
\end{enumerate}
We will also make the following assumptions:
\begin{enumerate}
\item[(13)]When an armed civilian fights back during Phase 1, he will kill the shooter with a probability $\tilde{p}_2$.  At the same time, he may also mistakenly kill another civilian with a probability of $\tilde{p}_f$.  We assume that these two events are independent and that  $0<\tilde{p}_f<\tilde{p}_2$.
\item[(14)]We assume $j_{K_i}\leq K_i\leq N'_i=N_2-(i-1)(e_2+m_2)$, where $N_2=N-n(e_1+m_1)$.
\end{enumerate}
For $i=1, \ldots, \max\{T_1, T_2\}$, we put $\tilde{p}_{2i}=[1-(1-\tilde{p}_2)^{j_{K_i}}][1-\tilde{p}_f/(N'_i-1)]^{j_{K_i}-1}$ if a civilian is armed and $\tilde{p}_{2i}=\tilde{p}_1$ if he is unarmed.  Then all results in Section~2.3 carry over here except that $\tilde{p}$ is replaced by $\tilde{p}_{2i}$ in all expressions  for $1\leq i\leq \max \{T_1, T_2\}$; see Table~\ref{table:sp6.1} and Table~\ref{table:sp6.2}. \\

\begin{table}[!th]
\begin{center}
\begin{tabular}{ccccc}
\hline
minute & $1$ &  $2$ & $\ldots$  & $T_2$   \\
\hline
fight & $p_s\tilde{p}_{21}$ & $p_s\tilde{p}_{22}\left(1-\frac{m_2}{N_2-e_2}\right)$ & $\ldots$ & $p_s\tilde{p}_{2T_2}\prod_{k=1}^{T_2-1}\left(1-\frac{m_2}{N'_k}\right)$     \\
hide & $p_s\left(1-\frac{m_2}{N_2-e_2}\right)$ & $p_s\prod_{k=1}^2 \left(1-\frac{m_2}{N'_k}\right)$ & $\ldots$ &  $p_s\prod_{k=1}^{T_2}\left(1-\frac{m_2}{N_k'} \right)$   \\
run &  $p_s \tilde{p}_r$ & $p_s\tilde{p}_r \left(1-\frac{m_2}{N_2-e_2}\right)$ & $\ldots$ & $p_s \tilde{p}_r \prod_{k=1}^{T_2-1} \left(1-\frac{m_2}{N_k'}\right)$   \\
\hline
\end{tabular}
\end{center}
\caption{Survival probability of a civilian of focus in a complex arena during Phase 2 when $T_1>T_2$, where $p_s$ is given by (\ref{eq:survprob}),  $N_k'=(N_2-e_2)-(k-1)(e_2+m_2)$, and the $i$-th  cell of the ``fight'' (respectively ``run'') row is the survival probability of the civilian when he hides for the first $(i-1)$ minutes and starts to fight back (respectively ``run'') at the beginning of the $i$-th minute.}
\label{table:sp6.1}
\end{table}

\begin{table}[!th]
\begin{center}
\begin{tabular}{ccccccc}
\hline
minute & $1$ &  $2$ & $\ldots$  & $T_1-1$  & $T_1$   & $T_1+1$\\
\hline
fight & $p_s\tilde{p}_{21}$ & $p_s\tilde{p}_{22} C_1$ & $\ldots$ & $p_s\tilde{p}_{2(T_1-1)}C_{T_1-2}$ & $p_s \tilde{p}_{2T_1}C_{T_1-1}$   & $0$ \\
hide & $p_s p_h C_1$ & $p_sC_2$ & $\ldots$ &  $p_sC_{T_1-1}$   & $0$ & $0$\\
run &  $p_s\tilde{p}_r$ & $p_s\tilde{p}_rC_1$ & $\ldots$ & $p_s \tilde{p}_rC_{T_1-2}$  & $p_s \tilde{p}_rC_{T_1-1}$  & $0$ \\
\hline
\end{tabular}
\end{center}
\caption{Survival probability of a civilian of focus in a complex arena when $T_1 \leq T_2$, where $p_s$ is given by (\ref{eq:survprob}),  $N_k'=\left(1-\frac{m_2}{N'_k} \right)=\left[1-\frac{m_2}{((N_2-e_2)-(k-1)(e_2+m_2))}\right]$,  $C_i=\prod_{k=1}^i\left(1-\frac{m_2}{N_k'}\right)=\prod_{k=1}^i\left[1-\frac{m_2}{((N_2-e_2)-(k-1)(e_2+m_2))}\right]$, and the $i$-th  cell of the ``fight'' (respectively ``run'') row is the survival probability of the civilian when he hides for the first $(i-1)$ minutes and starts to fight back (respectively ``run'') at the beginning of the $i$-th minute..}
\label{table:sp6.2}
\end{table}

\newpage
\section{Concluding remarks}
The past two decades have witnessed an increasing number of active shooter incidents in the US.  Several scholars have proposed analytical and agent-based simulation models to study these events.  However, no analytical model has taken civilian response into consideration.  Also,  most of these models focus on the expected number of civilian casualties.
We have investigated the optimal action of a civilian during an active shooter event where a civilian can fight,  hide, or run.  Unlike previous authors,  we incorporated civilian resistance into our models and focused on the survival probability of a civilian.  We considered three types of arenas: closed, open, and complex. For each type of arena,  we investigated the case where there is one possibly armed civilian and the case where there are multiple possibly armed civilians.   We have found the following results: 
\begin{enumerate}
\item[(i)]A civilian's chance of survival decreases over time,  irrespective of his action.
\item[(ii)]There are desperate situations in which a civilian should fight back even if he is unarmed.
\item[(iii)]Carrying a firearm does not always increase the civilian's chance of survival during an active shooter incident.
\item[(iv)]Carrying a firearm makes ``resistance'' more likely to be the optimal action of the civilian. 
\item[(v)] ``Hide'' might be the best action even if the civilian is armed. 
\item[(vi)]More armed civilians might increase a civilian's chance of survival during an active shooter event, though that is not always the case.
\end{enumerate}

The setup of our model is simpler than that of many agent-based models in the literature.  However, this should be a reason to dismiss an analytic model. In an analytic model approach,  the number of scenarios we need to consider grows exponentially when the number of factors increases.  If there are $n$ factors and each of them has a binary outcome, then we need to consider $2^n$ scenarios.  Therefore, we limited the number of factors in our model only to a few essential ones to make the analysis manageable. This is one inherent limitation of an analytic model.  However,  the conclusions reached using analytic models are analytically valid---an agent-based model can never achieve that. On the other hand, an agent-based model has the advantage that many factors can be included. Thus, they are equally important in the sense they can provide important insight empirically.

\section*{Conflict of interest}
The author has no conflict of interest to disclose.


\section*{References}

\begin{description}

\item{} Anklam, C., Kirby, A., Sharevski, F.~and Dietz, J.E.~(2015).  Mitigate active shooter impact; analysis for policy options based on agent/computer-based modeling. \emph{Journal of Emergency Management}~13(3), 201--216.

\item{} Blum, D.~and Jaworski, C. G. ~(2017). Spatial patterns of mass shootings in the United States, 2013–2014. In L., Leonard (Ed.), \emph{Environmental criminology} (Vol. 20, pp. 57–68). Emerald Publishing Limited. 

\item{} Cabrera, J. F.,~and Kwon,  R. ~(2018). Income inequality, household income, and mass shooting in the United States. \emph{Frontiers in Public Healt}~6, 294.

\item{} Cao,  L. , Mei, X.~and~Li, J.~(2024). Correlates of severity of mass public shootings in
the United States, 1966-–2022. \emph{Journal of Applied Security Research}, to appear, \url{DOI: 10.1080/19361610.2024.2314405}.

\item{} Dong, T.~(2019). Run, Hide, Fight—From the AR-15? A study of firearms depictions in active/mass shooter instructional videos. \emph{Journal of Applied Security Research}~14(3), 329--349

\item{} FBI~(2016).  Active shooter resources.  \url{https://www.fbi.gov/about/partnerships/office}\\
\url{-of-partner-engagement/active-shooter-resources}.

\item{} FBI~(2019).  Active shooter incidents 20-year review, 2000-2019.  \emph{Federal Bureau of Investigation}, \url{https://www.fbi.gov/file-repository/active shooter-incidents-20-year-}\\
\url{review-2000-2019-060121.pdf/view}.

\item{} FBI~(2023).  Active shooting incidents in the US in 2022. \emph{Federal Bureau of Investigation}, \url{https://www.fbi.gov/file-repository/active shooter-incidents-in-the-us-}\\
\url{2022-042623.pdf/view}.

\item{} FBI~(2024).  Active shooting incidents in the US in 2023. \emph{Federal Bureau of Investigation}, \url{https://www.fbi.gov/file-repository/2023-active-shooter-report-062124.pdf/view}

\item{} Fridel, E. E. ~(2021). Comparing the impact of household gun ownership and concealed carry legislation on the frequency of mass shootings and firearms homicide.  {Justice Quarterly}~38(5), 892--915.

\item{} Geller, L. B., Booty, M.~and Crifasi, C. K. ~(2021). The role of domestic violence in fatal mass shootings in the United States, 2014–2019.  \emph{Injury Epidemiology}, 8, 1–8. 

\item{} Gius, M. ~(2015). The impact of state and federal assault weapons bans on public mass shootings.  \emph{Applied Economics Letters}~22(4), 281--284.

\item{} Glass, P., Iyer, S., Lister-Gruesbeck, K. Schulz, N.~and Dietz, J.E.~(2018). Use of computer simulation  modeling to reduce the consequences of an active shooter event in a large event venue. \emph{2018 IEEE International Symposium on Technologies for Homeland Security}~1062--1065. 

\item{} Gunn, S., Luh, P.B., Lu, X., Hoteling, B.~(2017). Optimizing guidance for an active shooter event. \emph{2017 IEEE International Conference on Robotics and Automation}~4299--4304. 

\item{} Hayes, R. ~and Hayes, R. ~(2014). Agent-based simulation of mass shootings: determining how to limit the scale of a tragedy. \emph{Journal of Artificial Societies and Social Simulation}~17(2): 1--13.

\item{} Kaplan, E.H.~and Kress, M.~(2005). Operational effectiveness of suicide-bomber-detector schemes: a best-case analysis. \emph{Proceedings of the National Academy of Sciences}~102(29), 10399--10404.

\item{} Kingshott, B.F.~and McKenzie, D.G.~(2013).  Developing crisis management protocols in the context of school safety. \emph{Journal of Applied Security Research}~8(2), 222--245.

\item{} Kwon, R.,~and Cabrera, J. F. ~(2019).  Socioeconomic factors and mass shootings in the United States.  \emph{Critical Public Health}~29(2), 138--145.

\item{} Kress, M.~(2005). The effect of crowd density on the expected number of causalities in a suicide attack.  \emph{Naval Research Logistics}~52, 22--29. 

\item{} Lankford, A. ~(2015). Mass shooters in the USA, 1966--2010: differences between attackers who live and die.~\emph{Justice Quarterly}~32(2), 360--379.

\item{} Lankford, A. ~(2016).  Fame-seeking rampage shooters: initial findings and empirical predictions. ~\emph{Aggression and Violent Behavior}~27, 122--129. 

\item{} Lee, J.H.~(2013). School shooting in the US public schools: analysis through the eyes of an educator. \emph{Review of Higher Education and Self-learning}~6,(22), 88--120. 

\item{} Lee, J.Y., Dietz, E.~and Ostrowski, K.~(2018). Agent-based modeling for casualty rate assessment of large event active shooter incidents.  \emph{2018 Winter Simulation Conference}, 2737--2746.

\item{} Lei, X.~and MacKenzie, C.~(2024). Quantifying the risk of mass shootings at specific locations. \emph{Risk Analysis}, to appear, \url{https://doi.org/10.1111/risa.14197}.

\item{} Lemieux, F.~(2014).  Effect of gun culture and firearm laws on gun violence and mass shootings in the United States: a multi-level quantitative analysis.  \emph{International Journal of Criminal Justice Science}~9(1)., 74--93. 

\item{} Lin,  P.I., Fei, L., Barzman, D.~and Hossain, M.~(2018). What have we learned from the time trend of mass shooting in the US? \emph{PloS One}~13(10): e0204722.


\item{} Lu,  P., Li, Y., Wen, F.~and Chen, D.~(2023). Agent-based modeling of mass shooting case with the counterforce of policemen. \emph{Complex $\&$ Intelligent Systems}~9, 5093--5113.

\item{} Metzl, J.M.~and MacLeish, K.T.~(2015). Mental illness, mass shootings, and the politics of American firearms. \emph{American Journal of Public Health}~105(2): 240--249. 


\item{} Silva, J. R.~and Greene-Colozzi, E. A. ~(2021).  Mass shootings and routine activities theory: The impact of motivation, target suitability, and capable guardianship on fatalities and injuries.  \emph{Victims and Offenders}~16(4), 565–-586.

\item{} Soni, A.~and Tekin, E. ~(2020). How do mass shootings affect community wellbeing? Technical Report, National Bureau of Economic Research.

\item{} Sporrer, A.~(2020). Analytic models for active shooter incidents.  \emph{Technical Report of Naval Postgraduate School}. \url{https://apps.dtic.mil/sti/citations/AD1114719}.

\item{} Steward,  A.~(2017).  Active shooter simulations: An agent-based model of civilian response strategy.  \emph{Technical Report of Iowa State University}.  \url{https://doi.org/10.31274/etd-180810-5235}. 

\item{} The Violence Project~(2020).  Mass shooting database locations.  Technical Report.  \url{https://www.theviolenceproject.org/mass-shooter-database}.

\item{} Washburn, A.~and Kress, M.~(2009).  \emph{Combat Modeling}.  Springer: New York. 

\item{} Wiki (2023).  List of mass shootings in the United States, Wikipedia, \url{https://en.wikipedia.org/wiki/List_of_mass_shootings_in_the_United_States}.
\end{description}

\end{document}